\documentclass[11pt, a4paper]{amsart}

\usepackage[utf8]{inputenc}
\usepackage[T1]{fontenc}
\usepackage[english]{babel}
\usepackage{amsmath}
\usepackage{amsfonts}
\usepackage{ae}
\usepackage{units}
\usepackage{icomma}
\usepackage{color}
\usepackage{graphicx}
\usepackage{amsthm}
\usepackage{amssymb}
\usepackage{cancel}
\usepackage{fullpage}
\usepackage{upgreek}
\usepackage{type1cm}
\usepackage{eso-pic}
\usepackage[breaklinks,pdfpagelabels=false]{hyperref}

\newcommand{\bm}[1]{\mbox{\boldmath $ {#1} $}}

\renewcommand{\epsilon}{\varepsilon}

\newtheorem{theorem}{Theorem}[section]

\newtheorem{proposition}[theorem]{Proposition}

\theoremstyle{definition}
\newtheorem{rem}[theorem]{Remark}
\newtheorem{question}[theorem]{Question}

\numberwithin{equation}{section}
\numberwithin{theorem}{section}

\begin{document}

\title[The Hegselmann-Krause dynamics for equally spaced agents] 
{The Hegselmann-Krause dynamics for equally spaced agents}

\author{Peter Hegarty$^{1,2}$ \and Edvin Wedin$^{1,2}$} 
\address{$^1$Mathematical Sciences, Chalmers, 41296 Gothenburg, Sweden} 
\address{$^2$Mathematical Sciences, University of Gothenburg,  41296 Gothenburg, Sweden} 
\email{hegarty@chalmers.se}

\email{edvinw@student.chalmers.se}



\subjclass[2000]{} \keywords{}

\date{\today}

\begin{abstract} 
We consider the Hegselmann-Krause bounded confidence dynamics for $n$ equally
spaced opinions on the real line, with gaps equal to the confidence bound $r$, which we take to be $1$. We prove rigorous results on the evolution of this configuration, which confirm hypotheses previously made based on simulations for small values of $n$. Namely, for every $n$, the system evolves as follows: after every $5$ time steps, a group of $3$ agents become disconnected at either end and collapse to a cluster at the subsequent step. This continues until there are fewer than $6$ agents left in the middle, and these finally collapse to a cluster, if $n$ is not a multiple of $6$. In particular, the final configuration consists of $2 \lfloor n/6 \rfloor$ clusters of size $3$, plus one cluster in the middle of size $n \, ({\hbox{mod $6$}})$, if $n$ is not a multiple of $6$, and the number of time steps before freezing is $5n/6 + O(1)$. We also consider the dynamics for arbitrary, but constant, inter-agent spacings $d \in [0,\, 1]$ and present 
three main findings. Firstly we prove that the evolution is periodic also at 
some other, but not all, values of $d$, and present numerical evidence that for all $d$ something ``close'' to periodicity nevertheless holds. Secondly, we exhibit a value of $d$ at which the behaviour is periodic and the time to freezing is $n + O(1)$, hence slower than that for $d = 1$. Thirdly, we present numerical evidence that, as $d \rightarrow 0$, the time to freezing may be closer, in order of magnitude, to the diameter $d(n-1)$ of the configuration rather than the number of agents $n$. 
\end{abstract}

\maketitle



\setcounter{section}{0} 

\setcounter{equation}{0} 

\section{Introduction}\label{sect:intro}

The Hegselmann-Krause (HK) bounded confidence model of opinion dynamics, in its original one-dimensional setting introduced in \cite{HK}, works as follows. We have a finite number $n$ of agents, indexed by the integers $1,\, 2,\dots,\,n$. Time is measured discretely and the opinion of agent $i$ at time $t \in \mathbb{N} \cup \{0\}$ is represented by a real number $x_{t}(i)$, where the convention is that $x_{t}(i) \leq x_{t}(j)$ whenever $i \leq j$. There is a fixed parameter $r > 0$ such that the dynamics are given by
\begin{equation}\label{eq:update}
x_{t+1}(i) = \frac{1}{|\mathcal{N}_{t}(i)|} \sum_{j \in \mathcal{N}_{t}(i)} x_{t}(j),
\end{equation}
where $\mathcal{N}_{t}(i) = \{j : |x_{t}(j) - x_{t}(i)| \leq r \}$. As the dynamics are obviously unaffected by rescaling all opinions and the confidence bound $r$ by a common factor, we can assume without loss of generality that $r = 1$. 
\par Let $(x(1),\dots, \, x(n))$ be a vector of opinions. We say that agents $i$ and $j$ \emph{agree} if $x(i) = x(j)$. A maximal set of agents that agree is called a \emph{cluster}, and the number of agents in a cluster is called its \emph{size}. The configuration $\bm{x} = (x(1),\dots, \, x(n))$ is said to be \emph{frozen} if $|x(i) - x(j)| > 1$ whenever $x(i) \neq x(j)$. It is easy to see that, if $\bm{x}_t$ and $\bm{x}_{t+1}$ are related as in (\ref{eq:update}), then $\bm{x}_{t+1} = \bm{x}_t$ if and only if $\bm{x}_t$ is frozen. Thus once opinions obeying the HK-dynamics become frozen, they will remain so for all future time. 
\par Perhaps the most fundamental result about the HK-dynamics is that any configuration of opinions will freeze in finite time. There are multiple proofs of this in the literature, but the same fact is true for a wide class of models of which HK is just one particularly simple example, see \cite{C}. More interestingly, the time taken for a configuration to freeze is bounded by a universal function of the number $n$ of agents. Currently, the best upper bound is $O(n^3)$, due to \cite{BBCN}. An important open problem in the field is to find the optimal bound. 
\par It was noted in \cite{BBCN}, and even earlier in \cite{MBCF}, that one definitely cannot do better than an $O(n)$ bound. For suppose we start from the configuration $\bm{\mathcal{E}}_n = (1,\, 2,\dots, \, n)$, so opinions are equally spaced with gaps equal to the confidence bound. Then it is not hard to see that, as the configuration updates, if $i < n/2$ then the opinions of agents $i$ and $(n+1)-i$ will remain constant as long as $t < i$, while both will change at $t=i$. Hence, the time it takes for the configuration $\bm{\mathcal{E}}_n$ to freeze is at least $n/2$. 
\par Intuitively, $\bm{\mathcal{E}}_{n}$ seems like a good candidate for a configuration which converges as slowly as possible simply because, at the outset, opinions are placed as far apart as they can be while retaining an unbroken chain of influence. As it turns out, this intuition is badly wrong - in a companion paper \cite{WH} we will exhibit configurations of $n$ agents which, as $n \rightarrow \infty$, take time $\Omega (n^2)$ to freeze. Configurations of equally spaced agents nevertheless remain interesting for other reasons. Krause \cite{Kr} has observed, based on simulations for values of $n$ up to $100$ or so, that the configuration $\bm{\mathcal{E}}_n$ seems to evolve in a very regular manner. Our main result confirms, and makes completely precise, Krause's hypotheses. Before stating it, we introduce some graph-theoretic terminology. Let $(x(1),\dots, \, x(n))$ be a vector of opinions. We can define a {\em receptivity graph} $G$, whose nodes are the $n$ agents and where an edge is placed between agents $i$ and $j$ whenever 
$|x(i) - x(j)| \leq 1$. We say that agents $i$ and $j$ are \emph{connected} if they are in the same connected component of the receptivity graph. Observe that every connected component of $G$ is an interval of agents and that $i$ is disconnected from $i+1$ if and only if $x(i+1) > x(i) + 1$. We can now state our theorem:

\begin{theorem}\label{thm:main}
Let $n \geq 2$ be an integer, and write $n = 6k + l$ where $0 \leq l \leq 5$. Suppose that at $t=0$ we have the opinion vector $\bm{\mathcal{E}}_n$ and we let it evolve according to (\ref{eq:update}). Then the following occurs:
\par (i) after every fifth time step, a group of three agents will disconnect from either end of the receptivity graph and then collapse to a cluster in the subsequent time step. 
 \par (ii) the final, frozen configuration, will consist of $2k$ clusters of size $3$ with opinions distributed symmetrically about $\frac{n+1}{2}$ plus, if $l > 0$, one cluster of size $l$ with opinion $\frac{n+1}{2}$.
\par (iii) the configuration will freeze at time $t = 5k + \epsilon(l)$, where 
\begin{equation}\label{eq:remainder}
\epsilon(l) = \left\{ \begin{array}{lr} l-1, & {\hbox{if $l \in \{2,\, 3\}$}}, 
\\ l, &  {\hbox{if $l=1$}}, \\ l+1, & 
{\hbox{if $l \in \{0, \, 4, \, 5\}$}}. \end{array} \right.
\end{equation}
\end{theorem}

\begin{rem}\label{rem:kurz}
The formula for the freezing time for $n$ agents can be written as
\begin{equation}\label{eq:kurzformula}
T(n) = 1 + 5 \lfloor \frac{n+2}{6} \rfloor + \frac{1}{3} \left( \sqrt{3} \sin \left( \frac{2\pi (n-1)}{3} \right) - \cos \left( \frac{\pi (n-1)}{3} \right) - (-1) \right).
\end{equation}  
This formula is given in \cite{Ku}, but without proof.
\end{rem}
 
Theorem \ref{thm:main} will be proven in Section 2. In Section 3 we investigate more generally the evolution of a configuration of equally spaced agents, when the inter-agent spacing is an arbitrary number $d \in (0,\, 1]$. We will show that the evolution is periodic also for some other values of $d$, though not all, while numerical and heuristic evidence suggests nevertheless that something ``close'' to periodicity might hold for arbitrary $d$. We will show that, for a small range of values of $d$ slightly above $0.8$, the behaviour is periodic, with groups of four agents disconnecting after every eighth time step, which leads to a freezing time of $n + O(1)$. Thus $\bm{\mathcal{E}}_{n}$ does not even converge most slowly amongst equally spaced configurations. We conjecture, however, that any equally spaced configuration of $n$ agents will freeze in time $n + O(1)$, and present numerical evidence suggesting that, as $d \rightarrow 0$, the freezing time may be closer, in order of magnitude, to the diameter $d(n-1)$ of the configuration rather than the number $n$ of agents.  
  
\setcounter{equation}{0}

\section{Proof of Theorem \ref{thm:main}}\label{sect:pfmain}

For $n \in \mathbb{N} \cup \{0\}$ set
\begin{eqnarray}
\mathbb{R}_{*}^{n} = \{ (x(1),\dots, \, x(n)) \in \mathbb{R}^{n} : x(i) \leq x(j) \;
{\hbox{whenever}} \, i \leq j \}, \\
\mathbb{R}_{+}^{n} = \{ (x(1),\dots, \, x(n)) \in \mathbb{R}^{n} : x(i) \geq 0 \;
{\hbox{for each $i$}} \}.
\end{eqnarray}
If $n \geq 2$, there is a natural map $\phi: \mathbb{R}_{*}^{n} \rightarrow 
\mathbb{R}_{+}^{n-1}$ given by
\begin{equation}\label{eq:map}
\phi \left[ (x(1),\dots, \, x(n)) \right] = (x(2)-x(1),\dots, \, x(n)-x(n-1)).
\end{equation}
The update rule (\ref{eq:update}) can be written in matrix form as
\begin{equation}\label{eq:matrixA}
\bm{x}_{t+1} = A_{t} \, \bm{x}_t,
\end{equation}
where $\bm{x}_t, \bm{x}_{t+1} \in \mathbb{R}_{*}^{n}$ and $A_t$ is a 
row-stochastic $n \times n$ matrix. Note that the entries of $A_t$ depend only on the receptivity graph $G_t$ at time $t$. If $n \geq 2$ then, setting $\bm{y}_t = \phi(\bm{x}_t)$, we can just as well write the update rule as
\begin{equation}\label{eq:matrixB}
\bm{y}_{t+1} = B_t \, \bm{y}_t,
\end{equation}
where $\bm{y}_{t}, \bm{y}_{t+1} \in \mathbb{R}_{+}^{n-1}$ and $B_t$ is an
$(n-1) \times (n-1)$ matrix with non-negative entries, though the row sums will no longer equal one in general. As before, the entries of $B_t$ depend only on the graph $G_t$. We will find it more convenient to work with
(\ref{eq:matrixB}) rather than (\ref{eq:matrixA}), in other words to replace a vector $\bm{x}$ of opinions by a vector $\bm{y} = \phi(\bm{x})$ of gaps between opinions. Observe that $\phi(\bm{\mathcal{E}}_n) = (1, \, 1,\dots, \, 1)$, a vector which we denote $\mathbf{1}_{n-1}$. More generally, if the receptivity graph corresponding to $\bm{x}$ is connected, then the entries of $\phi(\bm{x})$ are bounded above by one.
\par To get a feeling for Theorem \ref{thm:main}, we look at $n=11$ as an example. What is important is what happens during the first five time steps. At $t = 0$ it is obvious that
\begin{eqnarray}
\bm{y}_0 = (1,\, 1,\, 1,\, 1,\, 1,\, 1,\, 1,\, 1,\, 1, \, 1), \\
E(G_0) = \{\{i, \, i+1\} : 1 \leq i \leq 10\}, \\
B_{0}(i,j) = \left\{ \begin{array}{lr} 1/6, & {\hbox{if $i=j=1$ or $i=j=10$}}, 
\\ 0, & {\hbox{if $|i-j| > 1$}}, \\ 1/3, & {\hbox{otherwise}}. \end{array} \right.
\end{eqnarray}
Thus, 
\begin{equation}\label{eq:y1}
\bm{y}_1 = B_0 \, \bm{y}_0 = (0.5, \, 1, \, 1, \, 1, \, 1, \, 1, \, 1, \, 1, \, 1, \, 0.5),
\end{equation}
from which it is in turn clear that $G_1 = G_0$ and $B_1 = B_0$. The entries of all subsequent vectors $\bm{y}_t$ will of course all be rational numbers, but we will write approximations to four decimal places so as to make it easier to keep track of magnitudes. At the next time step we have
\begin{equation}\label{eq:y2}
\bm{y}_2 = B_1 \, \bm{y}_1 = B_{0}^{2} \, \bm{y}_0 \approx (0.4167, \, 0.8333, \, 1, \, 1, \, 1, \, 1, \, 1, \, 1, \, 0.8333, \, 0.4167).
\end{equation}
Thus $G_2 = G_0$ and $B_2 = B_0$ still and so
\begin{equation}\label{eq:y3}
\bm{y}_3 = B_{0}^{3} \, \bm{y}_0 \approx (0.3472, \, 0.75, \, 0.9444, \, 1, \, 1, \, 1, \, 1, \, 0.9444, \, 0.75, \, 0.3472).
\end{equation}
Still $G_3 = G_0$ and so 
\begin{equation}\label{eq:y4}
\bm{y}_4 = B_{0}^{4} \, \bm{y}_0 \approx (0.3079, \, 0.6806, \, 0.8981, \, 0.9815, \, 1, \, 1, \, 0.9815, \, 0.8981, \, 0.6806, \, 0.3079).
\end{equation}
Now finally something happens. Since $\bm{y}_{4}(1) + \bm{y}_{4}(2) < 1$, agents $1$ and $3$ are now connected, and similarly with agents $9$ and $11$. Thus 
$E(G_4) = E(G_0) \cup \{\{1, \, 3\}, \{9, \, 11\}\}$ which leads to
\begin{equation}\label{eq:b4}
B_{4} = \left( \begin{array}{cccccccccc} 
0&0&0&0&0&0&0&0&0&0 \\
\frac{1}{12} & \frac{1}{6} & \frac{1}{4} & 0 & 0 & 0 & 0 & 0 & 0 & 0 \\
\frac{1}{4} & \frac{1}{2} & \frac{5}{12} & \frac{1}{3} &0&0&0&0&0&0 \\
0&0&\frac{1}{3}&\frac{1}{3}&\frac{1}{3}&0&0&0&0&0 \\
0&0&0&\frac{1}{3}&\frac{1}{3}&\frac{1}{3}&0&0&0&0 \\
0&0&0&0&\frac{1}{3}&\frac{1}{3}&\frac{1}{3}&0&0&0 \\
0&0&0&0&0&\frac{1}{3}&\frac{1}{3}&\frac{1}{3}&0&0 \\
0&0&0&0&0&0&\frac{1}{3}&\frac{5}{12}&\frac{1}{2}&\frac{1}{4} \\
0&0&0&0&0&0&0&\frac{1}{4}&\frac{1}{6}&\frac{1}{12} \\
0&0&0&0&0&0&0&0&0&0 
\end{array} \right).
\end{equation}
Hence, 
\begin{equation}\label{eq:y5}
\bm{y}_5 = B_{4} \, \bm{y}_4 \approx (0, \, 0.3636, \, 1.1186, \, 0.9599, \, 0.9938, \, 0.9938, \, 0.9599, \, 1.1186, \, 0.3636, \, 0).
\end{equation}
Here $\bm{y}_{5}(1) = \bm{y}_{5}(10) = 0$ exactly, which means that agents $1$ and $2$ agree, as do agents $10$ and $11$. Moreover, $\bm{y}_{5}(3) = \bm{y}_{5}(8) > 1$, which means that agent $3$ (resp. $8$) has become disconnected from agent $4$ (resp. $9$). This is in accordance with part (i) of Theorem \ref{thm:main}. It is easy to continue and check that agents $1, \, 2, \, 3$ will collapse to a cluster at $t=6$, as will agents $9, \, 10, \, 11$, whereas the remaining agents $4, \, 5, \, 6, \, 7, \, 8$ will collapse to a cluster at $t=11$, in accordance with (\ref{eq:remainder}). 

We can now give a rough outline of the proof of Theorem \ref{thm:main}. Given $n \geq 11$, we proceed as follows:
\\
\\
{\sc Step 1:} Show that the receptivity graph $G_t$, and hence the transition matrix $B_t$, is constant for $t = 0, \, 1, \, 2, \, 3$, whereas $G_4$ contains the two extra edges $\{1, \, 3\}$ and $\{n-2, \, n\}$. Thus $\bm{y}_5 = B_4 \, B_{0}^{4} \, \bm{y}_0$. Then show that $\bm{y}_{5}(3) = \bm{y}_{5}(n-3) > 1$, which implies that three agents become disconnected at each end. If $\bm{y}_{5}(1) + \bm{y}_{5}(2) < 1$, then each group of three agents will collapse to a cluster at $t = 6$. 
\\
\\
Note that we have already completed this step. The calculations above verify it for $n = 11$ and it is clear that, if we then increase $n$, it will not affect how the opinions of the first or last four agents evolve over the first five time steps. Indeed, the value of $n$ cannot be ``felt'' before every agent has changed their opinion at least once which, as we remarked earlier, will not happen 
while $t < n/2$. 
\\
\\
{\sc Step 2:} Thus at $t=5$, three agents break free from each end of the configuration. This leaves us with $n-6$ agents and a corresponding vector of 
gaps $\overline{\bm{y}}_{5} \in \mathbb{R}_{+}^{n-7}$. We now reset time to zero and consider $\overline{\bm{y}}_{5}$ as the new initial configuration. For example, with $n = 11$ we have, by (\ref{eq:y5}), 
\begin{equation}
\overline{\bm{y}}_{5} \approx (0.9599, \, 0.9938, \, 0.9938, \, 0.9599).
\end{equation}
The entries of $\overline{\bm{y}}_{5}$ will lie in $(0, \, 1]$. The idea is to show that they lie sufficiently close to $1$ such that, if still $n - 6 \geq 11$, then the evolution of the receptivity graph over the next five time steps will be exactly the same as if all entries equalled $1$.
To complete the proof, we have to be able to iterate this procedure, and finally verify the theorem directly for $n \leq 10$ - indeed, we need to verify for these values of $n$ that the behaviour is unaffected if the starting values in $\bm{y}_0$ are all in $(0, \, 1]$ and sufficiently close to $1$. 
\\
\par To complete Step 2, it is convenient to extend the HK-model to an infinite sequence of agents, more precisely to a well-ordered infinite sequence so that geometrically there is an agent furthest to the left. Indeed, (\ref{eq:matrixA}) and (\ref{eq:matrixB}) make perfect sense if we regard $\bm{x}_t$ and $\bm{y}_t$ as elements of $\mathbb{R}^{\infty}$, the vector space consisting of all infinite, well-ordered sequences of real numbers, and $A_t$, $B_t$ as appropriate linear operators on this space. Let $\bm{\mathcal{E}}_{\infty} = (1, \, 2, \dots)$ denote the element of $\mathbb{R}^{\infty}$ representing a sequence of equally spaced agents with gaps of one. The obvious analogue of Theorem \ref{thm:main} would be the following:

\begin{theorem}\label{thm:infmain}
The evolution of the configuration $\bm{\mathcal{E}}_{\infty}$ under (\ref{eq:update}) is periodic, namely after every fifth time step a group of three agents will disconnect on the left and then collapse to a cluster at the subsequent time step.
\end{theorem}

Our strategy will be to first prove Theorem \ref{thm:infmain} and then argue that the behaviour is essentially unaffected when we go back to finite sequences. 
\\
\par To begin with, we define precisely the machinery we need, using the standard notation and terminology of functional analysis. 
An element of $\mathbb{R}^{\infty}$ will be denoted by a well-ordered sequence $\bm{x} = (x(i))_{i = 1}^{\infty}$ of real numbers. Recall that $l^{\infty}$ denotes the subspace of $\mathbb{R}^{\infty}$ consisting of bounded sequences. It is a Banach space with norm $||\bm{x}|| = \sup_{i} |x(i)|$. We let 
$\mathbf{1}_{\infty} = (1, \, 1, \dots)$ denote the element of $l^{\infty}$ 
consisting entirely of ones. 
For a linear operator $T: l^{\infty} \rightarrow l^{\infty}$, its norm is defined
as $||T|| = \sup_{||\bm{x}|| = 1} ||T(\bm{x})||$.
One says that $T$ is bounded if $||T|| < \infty$ and $B(l^{\infty})$ denotes the Banach space of all bounded linear operators on $l^{\infty}$. Let $B = (b(i,\,j))_{i,\,j=1}^{\infty}$ be a doubly-infinite matrix and set 
\begin{equation}\label{eq:sup}
s = \sup_{i} \sum_{j=1}^{\infty} |b(i,\,j)|. 
\end{equation}
If $s$ is finite then the map $T$ given by
\begin{equation}
(T(\bm{x}))(i) = \sum_{j=1}^{\infty} b(i,\,j) \, x(j), \;\;\;\; \bm{x} \in l^{\infty},
\end{equation}
is a well-defined element of $B(l^{\infty})$ of norm $s$. The map can be written as a matrix product $T(\bm{x}) = B\bm{x}$, when $\bm{x}$ is written as an
infinite column. 

We now consider a specific collection of operators $\widetilde{B}_t$, $t = 0,\dots,\,4$ defined by matrices satisfying (\ref{eq:sup}). In all cases, the elements $b(i,\,j)$ will be non-negative and there will be only finitely many non-zero entries in each row, i.e.: for each $i$ we have $b(i,\,j) = 0$ for all $j \gg_{i} 0$. Set
\begin{eqnarray}
\widetilde{B}_{0}(i,\,j) = \left\{ \begin{array}{lr} 1/6, & {\hbox{if $(i,\,j) = (1,\,1)$}}, \\
1/3, & {\hbox{if $|i-j| \leq 1$ and $(i,\,j) \neq (1,\,1)$}}, \\
0, & {\hbox{otherwise}}, \end{array} \right.
\\
\widetilde{B}_3 = \widetilde{B}_2 = \widetilde{B}_1 = \widetilde{B}_0, \\
\widetilde{B}_{4}(i,\,j) = \left\{ \begin{array}{lr} B_{4}(i,\,j), & {\hbox{if $i \leq 3$ and
$j \leq 10$}}, \\ 1/3, & {\hbox{if $i > 3$ and $|i-j| \leq 1$}}, \\
0, & {\hbox{otherwise}}. \end{array} \right. \;\;\;\;
(B_4 \; {\hbox{as in (\ref{eq:b4})}}).
\end{eqnarray}
Let $S_3: l^{\infty} \rightarrow l^{\infty}$ be a threefold leftward shift, i.e.:
\begin{equation}\label{eq:shift}
(S_3 \bm{x})(i) = x(i+3), \;\;\; \forall \; i \geq 1,
\end{equation}
and finally let $\mathcal{T}: l^{\infty} \rightarrow l^{\infty}$ be the
composition
\begin{equation}\label{eq:Top}
\mathcal{T} = S_3 \circ \widetilde{B}_4 \circ (\widetilde{B}_0)^4.
\end{equation}
The point is that, firstly, the evolution under (\ref{eq:update}) of an infinite sequence of equally spaced agents is described for $t \leq 4$ by 
\begin{equation}\label{eq:matrixBtilde}
\bm{y}_{t+1} = \widetilde{B}_t \, \bm{y}_t, \;\;\;\; \bm{y}_0 = \mathbf{1}_{\infty}.
\end{equation}
Secondly, at $t=5$, the first three agents become disconnected.  
Hence, the map $\mathbf{1}_{\infty} \mapsto \mathcal{T} \mathbf{1}_{\infty}$ 
describes what happens if we take a sequence of equally spaced agents, run the HK-dynamics over $5$ time steps, remove the first three agents which have become disconnected from the others, and reset time to zero. These assertions follow from (\ref{eq:y1})-(\ref{eq:y4}) and (\ref{eq:y5}), together with the observation that, whenever $\bm{y}_0$ is a multiple of $\mathbf{1}_{\infty}$, so that all its entries are equal, then $\bm{y}_t (i) = \bm{y}_0 (i)$ for all $i \geq 6$ and all $t \leq 5$. Now we need to show two things:
\\
\\
{\sc Claim 1:} If $||\bm{y}|| \leq 1$ and $||\mathbf{1}_{\infty} - \bm{y}|| < \varepsilon$ for some sufficiently small $\varepsilon > 0$, then $\bm{y} \mapsto \mathcal{T} \bm{y}$ still describes, as above, the evolution over $5$ time steps of a sequence of agents with initial inter-agent spacings given by $\bm{y}$.
\\
\\
{\sc Claim 2:} For all $n \in \mathbb{N}$, $||\mathcal{T}^n \mathbf{1}_{\infty} - \mathbf{1}_{\infty}||$ is sufficiently small so that Claim 1 can be applied. 
\\
\par
Verifying these two claims will not immediately allow us to complete Step 2 above. Given a finite sequence of $n$ equally spaced agents, these two claims allow us to iterate as in Step 2 up as far as $t \approx n/2$, i.e.: as long as agents in the middle have not yet been affected, because up to that point the evolution of either half of the finite sequence is exactly the same as for the corresponding initial segment of the infinite sequence. In order to continue the iteration beyond this time, we will use the precise quantitative estimates obtained in the proofs of the two claims below. Basically, the idea is that the entries of
$\mathcal{T}^n \mathbf{1}_{\infty}$ remain much closer to one than is needed for Claim 1 to hold, so that even though the entries of the finite and infinite sequences diverge when $t > n/2$, the accumulated divergence never becomes so large so that we cannot apply Claim 1. The precise argument will follow the verification of the two claims.
\\
\par
Claim 1 is actually a statement about the operators $\widetilde{B}_t$, $0 \leq t \leq 4$. We prove the following:

\begin{proposition}\label{prop:claim1}
Let $\bm{y}_0 \in l^{\infty}$ represent the gaps between consecutive agents in an infinite sequence of agents, indexed by $1,\,2,\dots$, and let $G_0$ be the corresponding receptivity graph. For each $t = 0,\,1,\dots,\,4$, let $\bm{y}_{t+1} = \widetilde{B}_t \, \bm{y}_t$, and let $G_{t+1}$ be the corresponding graph. If $||\bm{y}_0|| \leq 1$ and $||\mathbf{1}_{\infty} - \bm{y}_0|| < \frac{7}{79}$ then the following hold: 
\par (i) for each $0 \leq t \leq 3$, $G_t$ contains exactly the edges
$\{i,\,i+1\}, \, i \in \mathbb{N}$, 
\par (ii) $G_4$ contains the edges of $G_3$ plus the additional edge $\{1,\,3\}$,
\par (iii) $G_5$ contains all the edges of $G_4$ except for the edge $\{3,\,4\}$. Moreover, $\bm{y}_{5}(1) = 0$.
\\
Hence Claim 1 holds with $\varepsilon = \frac{7}{79}$. 
\end{proposition}

\begin{proof}
We already know that (i)-(iii) hold when $\bm{y}_0 = \mathbf{1}_{\infty}$. Now take $\bm{z}_0 = \frac{72}{79} \mathbf{1}_{\infty}$. As remarked earlier, since $\bm{z}_0$ is a multiple of $\mathbf{1}_{\infty}$, all its entries from the sixth onwards will remain unchanged for $t \leq 5$. So, when considering the evolution of the graph $G_t$, it suffices to consider the first five entries of each vector $\bm{z}_t$. Since every entry of $\bm{z}_0$ lies between $1/2$ and $1$, it is immediate that $G_0$ is as claimed in (i). We have $\bm{z}_t = \frac{72}{79} \bm{y}_t$ for all $t$, where the $\bm{y}_t$ are as in (\ref{eq:y1})-(\ref{eq:y4}) and (\ref{eq:y5}) except that we have an infinite sequence of ones from the sixth position onwards. Numerically, 
\begin{equation}\label{eq:z1}
\bm{z}_1 \approx (0.4557, \, 0.9114, \, 0.9114, \, 0.9114, \, 0.9114, 
\, \dots), 
\end{equation}
\begin{equation}\label{eq:z2}
\bm{z}_2 \approx (0.3791, \, 0.7594, \, 0.9114, \, 0.9114, \, 0.9114, \, \dots), 
\end{equation}
\begin{equation}\label{eq:z3}
\bm{z}_3 \approx (0.3165, \, 0.6835, \, 0.8608, \, 0.9114, \, 0.9114, \, \dots),
\end{equation}
\begin{equation}\label{eq:z4}
\bm{z}_4 \approx (0.2806, \, 0.6203, \, 0.8186, \, 0.8945, \, 0.9114, \, \dots), 
\end{equation}
\begin{equation}\label{eq:z5}
\bm{z}_5 \approx (0, \, 0.3314, \, 1.0195, \, 0.8748, \, 0.9058, \, \dots).
\end{equation}
Note that $\bm{z}_3(1) + \bm{z}_{3}(2) = 1$ exactly, as may be readily checked. 
For sequences $\bm{u}, \bm{v} \in l^{\infty}$, write $\bm{u} < \bm{v}$ if
$u(i) < v(i)$ for every $i$ and observe that, if  
$M$ is a doubly-infinite matrix satisfying (\ref{eq:sup}), having non-negative entries and at least one non-zero entry in each row, then 
$\bm{u} < \bm{v} \Rightarrow M\bm{u} < M\bm{v}$. Now if $||\bm{w}_0|| \leq 1$ and $||\mathbf{1}_{\infty} - \bm{w}_0|| < \frac{7}{79}$, then $\bm{z}_0 < \bm{w}_0 \leq \bm{1}_{\infty}$. Hence $\bm{z}_t < \bm{w}_t \leq \bm{y}_t$ would hold for all $t \leq 5$, but for the fact that the first entry of each vector is zero at $t=5$. By comparing (\ref{eq:y1})-(\ref{eq:y4}), (\ref{eq:y5}) with (\ref{eq:z1})-(\ref{eq:z5}), and noting in particular that both $\bm{w}_{3}(1) + \bm{w}_{3}(2) > 1$ and $\bm{w}_{5}(3) > 1$, it follows immediately that (i)-(iii) all hold for the $\bm{w}_t$. 
\end{proof}

We now turn to Claim 2. First of all, let us write out $\mathcal{T}$ explicitly as a doubly-infinite matrix. The upper-left $3 \times 11$ block is
\begin{equation}\label{eq:Tblock}
\left( \begin{array}{ccccccccccc} 
\frac{47}{972}&\frac{59}{486}&\frac{5}{27}&\frac{17}{81}&\frac{5}{27}&\frac{10}{81}&\frac{5}{81}&\frac{5}{243}&\frac{1}{243}&0&0 \\
\frac{1}{54}&\frac{5}{81}&\frac{10}{81}&\frac{5}{27}&\frac{17}{81}&\frac{5}{27}&\frac{10}{81}&\frac{5}{81}&\frac{5}{243}&\frac{1}{243}&0 \\
\frac{1}{243}&\frac{5}{243}&\frac{5}{81}&\frac{10}{81}&\frac{5}{27}&\frac{17}{81}&\frac{5}{27}&\frac{10}{81}&\frac{5}{81}&\frac{5}{243}&\frac{1}{243}
\end{array} \right). 
\end{equation}
Every other entry in the first three rows is zero, and every row from the fourth onwards is just a rightward shift of the third row, i.e.:
\begin{equation}\label{eq:Trest}
\mathcal{T}(i,\, j) = \left\{ \begin{array}{lr} 0, & {\hbox{if $i \leq 3$ and
$j \geq 12$}}, \\ 0, & {\hbox{if $i \geq 4$ and $j=1$}}, \\
\mathcal{T}(i-1,\, j-1), & {\hbox{if $i \geq 4$ and $j \geq 2$}}. \end{array}
\right.
\end{equation}
Note that the sum of the entries in every row from the third onwards equals 
$1$. Let 
\begin{eqnarray}
a = \sum_{j=1}^{9} \mathcal{T}(1,\,j) = \frac{311}{324} \approx 0.9599, \\
b = \sum_{j=1}^{10} \mathcal{T}(2,\,j) = \frac{161}{162} \approx 0.9938.
\end{eqnarray}

\begin{proposition}\label{prop:claim2}
Let $\bm{y}_0 := \bm{1}_{\infty}$ and for all $\tau \geq 0$, $\bm{y}_{\tau+1} := 
\mathcal{T} \bm{y}_{\tau}$. Then 
\par (i) $\bm{y}_{\tau+1} \leq \bm{y}_{\tau}$ for all $\tau$.
\par (ii) Let $\gamma := 0.1117$. Then for all $\tau \geq 1$,
\begin{eqnarray}
1-a \leq 1 - \bm{y}_{\tau}(1) < (1+\gamma)(1-a), \label{eq:bound1}\\
1-b \leq 1 - \bm{y}_{\tau}(2) < \gamma (1-a) + (1+\gamma)(1-b), \label{eq:bound2}\\
1 - \bm{y}_{\tau}(i) < \gamma^{i-2} (1-a), \;\,\; \forall \, i \geq 3. \label{eq:bound3}
\end{eqnarray}
$\;\;$ (iii) the sequence $(\bm{y}_{\tau})_{\tau=0}^{\infty}$ converges in $l^{\infty}$
to a fixed point of $\mathcal{T}$. 
\end{proposition}

\begin{rem}
The reason for using $\tau$ instead of $t$ is that multiplication by $\mathcal{T}$ is to be thought of as representing the evolution of a configuration of agents over $5$ time steps. Hence, one can informally imagine that ``$\tau = 5t$''.
\end{rem}

\begin{proof}
Since the sum of the entries in every row of the matrix of $\mathcal{T}$ is at most one it is immediate that $\mathcal{T}(\bm{1}_{\infty}) \leq \mathbf{1}_{\infty}$. Since $\mathcal{T}$ has non-negative entries, it follows by induction that $\bm{y}_{\tau+1} \leq \bm{y}_{\tau}$ for every $\tau$, which proves (i). Note that (iii) follows from (i) and (ii), so it remains to prove (ii). 
\par The inequalities in (\ref{eq:bound1})-(\ref{eq:bound3}) are obviously satisfied when $\tau=1$. Hence, by (i), the left-hand inequalities in (\ref{eq:bound1}) and (\ref{eq:bound2}) will be satisfied for all $\tau \geq 1$. For the right-hand inequalities, we proceed
by induction on $\tau$. First consider $i = 1$. We have
\begin{equation*}
\bm{y}_{\tau+1}(1) = \sum_{j=1}^{9} \mathcal{T}(1,\,j) \, \bm{y}_{\tau}(j) \Rightarrow 
1 - \bm{y}_{\tau+1}(1) = (1-a) + \sum_{j=1}^{9} \mathcal{T}(1,\,j) \, (1 - \bm{y}_{\tau}(j)).
\end{equation*}
Assuming (\ref{eq:bound1})-(\ref{eq:bound3}) hold at step $\tau$, 
it follows that 
(\ref{eq:bound1}) holds at step $\tau+1$
if and only if $f_{1}(\gamma) \leq 0$, where
\begin{equation}\label{eq:row1}
f_{1}(\gamma) = (1-a)\left[-\gamma + \mathcal{T}(1,\,1)\,(1+\gamma) + 
\mathcal{T}(1,\,2) \, \gamma + 
\sum_{j=3}^{9} \mathcal{T}(1,\,j) \, \gamma^{j-2} \right] + 
(1-b)\, \mathcal{T}(1,\,2) \, (1+\gamma).
\end{equation}
We checked with \texttt{Matlab} that $f_{1}(\gamma) = 0$ has two solutions in the interval $[0,\,1]$, at $\gamma_1 \approx 0.1116$ and $\gamma_{2} \approx 0.9624$, and
that $f_{1}(\gamma) < 0$ for $\gamma \in (\gamma_1, \gamma_2)$. Thus (\ref{eq:bound1}) holds also at step $\tau+1$, provided $\gamma = 0.1117$. 
\par Next consider $i = 2$. Analogous calculations lead to the requirement that
$f_{2}(\gamma) \leq 0$, where
\begin{equation}\label{eq:row2}
f_{2}(\gamma) = (1-a)\left[ - \gamma + \mathcal{T}(2,\,1) \, (1+\gamma) + \mathcal{T}(2,\,2) \, \gamma + \sum_{j=3}^{10} \mathcal{T}(2,\,j) \, \gamma^{j-2} \right] + 
(1-b)\left[ - \gamma + \mathcal{T}(2,\,2) \, (1+\gamma) \right].
\end{equation}
One checks that $f_{2}(\gamma) = 0$ has one solution in $[0,\,1]$, at $\gamma_3 \approx 0.03$, and that $f_{2}(\gamma) < 0$ for all $\gamma \in (\gamma_3,\,1]$. Thus (\ref{eq:bound2}) holds also at step $\tau+1$, provided $\gamma = 0.1117$. 
\par For $i = 3$, we are led to the requirement that $f_{3}(\gamma) \leq 0$, where
\begin{equation}\label{eq:row3}
f_{3}(\gamma) = (1-a)\left[ - \gamma + \mathcal{T}(3,\,1) \, (1+\gamma) + \mathcal{T}(3,\,2) \, \gamma + 
\sum_{j=3}^{11} \mathcal{T}(3,\,j) \, \gamma^{j-2} \right] + (1-b) \, \mathcal{T}(3,\,2) \, (1+\gamma).
\end{equation}
One checks that $f_{3}(\gamma) < 0$ for all $\gamma \in (\gamma_4, \, \gamma_5)$, where $\gamma_4 \approx 0.008$ and $\gamma_5 \approx 0.9965$. 
\par For $i=4$, using (\ref{eq:Trest}) we are led to the requirement that $f_{4}(\gamma) \leq 0$ where
\begin{equation}\label{eq:row4} 
f_{4}(\gamma) = (1-a)\left[ -\gamma^{2} + \mathcal{T}(3,\,1) \, \gamma + 
\sum_{j=2}^{11} \mathcal{T}(3,\,j) \, \gamma^{j-1} \right] + (1-b) \, \mathcal{T}(3,\,1) \, (1+\gamma).
\end{equation}
One checks that $f_{4}(\gamma) < 0$ for all $\gamma \in (\gamma_6, \, \gamma_7)$, where $\gamma_6 \approx 0.0430$ and $\gamma_7 \approx 0.9996$. 
\par Finally, every $i \geq 5$ will lead to the condition that $f_{5}(\gamma) \leq 0$ where
\begin{equation}\label{eq:row5}
f_{5}(\gamma) = - \gamma^{2} + \sum_{j=1}^{11} \mathcal{T}(3,\, j) \, \gamma^{j-1}.
\end{equation}
One checks that $f_{5}(\gamma) < 0$ for all $\gamma \in (\gamma_8, \, 1]$, where
$\gamma_8 \approx 0.0786$. This completes the induction step.
\end{proof}
Claim 2 follows from Proposition \ref{prop:claim2}. Indeed, for each $i \geq 1$, let $g_i = g_{i}(a,b,\gamma)$ be the bound on the appropriate right-hand side in (\ref{eq:bound1})-(\ref{eq:bound3}). Since the $g_i$ form a decreasing sequence, we have proven that $||\mathcal{T}^{n} \bm{1}_{\infty} - \bm{1}_{\infty}|| < g_1 \approx 0.0446 < \frac{7}{79}$ for all $n$. Together with Proposition \ref{prop:claim1}, this already suffices to prove Theorem \ref{thm:infmain}.
\par Turning to finite sequences of agents, we are now ready to complete the proof of our main result.

\emph{Proof of Theorem \ref{thm:main}}.
Let $\bm{x}_0$ be any finite or infinite vector of equally spaced opinions, and $\bm{y}_0$ the corresponding vector of gaps. We start with a couple of observations:
\\
\\
(a) Suppose $\bm{x}_0$ is finite of length $n$, and thus $\bm{y}_0$ has length $n-1$. Given the updating rule (\ref{eq:update}), at all times $t$ we will have
$x_{t}(i) = x_{t}(n+1-i)$ for $i = 1,\dots, \, \lfloor \frac{n+1}{2} \rfloor$ 
and similarly $y_{t}(i) = y_{t}(n-i)$ for $i = 1,\dots, \, \lfloor \frac{n}{2} \rfloor$. Hence in order to understand how the configuration evolves, it suffices to understand the vectors $\left(y_{t}(1),\dots,\, y_{t}\left(\lfloor \frac{n}{2} \rfloor\right)\right)$ for all $t$. 
\\
\\
(b) Suppose the receptivity graph $G_t$ is as in Proposition \ref{prop:claim1} for all $t \leq 5$. Then, for any $k \in \mathbb{N}$ in the case of infinite vectors, or for any $k \leq n/2$ in the case of finite vectors, the values of $y_{t}(i)$ for all $0 \leq t \leq 5$, $1 \leq i \leq k$, only depend on $y_{0}(j)$, for $1 \leq j \leq k+5$.
\\
\par
Now let $n \geq 11$ be given. Write $\lfloor n/2 \rfloor := r$. 
For each $0 \leq \tau \leq m$, where $m$ is a bound to be determined in a moment, we define a sequence of ``replacement'' operators $\mathcal{R}_{\tau}: l^{\infty} \rightarrow l^{\infty}$ as follows:
\begin{equation}\label{eq:replace}
\left\{ \begin{array}{lr} (\mathcal{R}_{\tau} \bm{x})(i) = x(i), & {\hbox{if $i \leq r-3\tau$ or $i > r-3\tau+5$}}, \\ (\mathcal{R}_{\tau} \bm{x}(r-3\tau+j) = x(r-3\tau-j), & 
{\hbox{if $1 \leq j \leq 5$ and $n$ is even}}, \\
(\mathcal{R}_{\tau} \bm{x})(r-3\tau+j) = x(r-3\tau+j+1), & {\hbox{if $1 \leq j \leq 5$ and
$n$ is odd.}} \end{array} \right.
\end{equation}
The definition makes sense as long as $\tau \leq m$, where 
\begin{equation}\label{eq:mbound}
m = \left\{ \begin{array}{lr} \lfloor \frac{r-6}{3} \rfloor, & {\hbox{if $n$ is even}}, \\ \lfloor \frac{r-5}{3} \rfloor, & {\hbox{if $n$ is odd.}} \end{array} \right.
\end{equation}
Note that each $\mathcal{R}_{\tau}$ is linear and of norm one, since every 
entry in $\mathcal{R}_{\tau} \bm{x}$ is also an entry in $\bm{x}$. 
\par Set $\bm{y}_0 = \bm{z}_0 = \bm{1}_{\infty}$. For each $0 \leq 
\tau \leq m$, let  
$\Theta_{\tau} := \mathcal{T} \circ \mathcal{R}_{\tau}$ and define inductively 
\begin{equation}\label{eq:yzdefs}
\bm{y}_{\tau+1} = \mathcal{T} \bm{y}_{\tau}, \;\;\; \bm{z}_{\tau+1} = \Theta_{\tau} \, \bm{z}_{\tau}, 
\;\;\; \delta_{\tau} := ||\bm{z}_{\tau} - \bm{y}_{\tau}||.
\end{equation}
We are interested in bounding the $\delta_{\tau}$. Clearly, $\delta_0 = 0$. 
Since $||T|| = 1$ we have
\begin{equation}\label{eq:Tbort}
\delta_{\tau+1} = ||\mathcal{T}(\mathcal{R}_{\tau} \, \bm{z}_{\tau}) - \mathcal{T}(\bm{y}_{\tau})||
\leq ||\mathcal{R}_{\tau} (\bm{z}_{\tau}) - \bm{y}_{\tau}||. 
\end{equation}
Next, by the triangle inequality and the properties of $\mathcal{R}_{\tau}$, 
\begin{eqnarray*}
||\mathcal{R}_{\tau} (\bm{z}_{\tau})  - \bm{y}_{\tau}|| = ||\mathcal{R}_{\tau}(\bm{z}_{\tau}) - 
\mathcal{R}_{\tau} (\bm{y}_{\tau}) + \mathcal{R}_{\tau} (\bm{y}_{\tau}) - \bm{y}_{\tau} || 
\leq ||\mathcal{R}_{\tau} (\bm{z}_{\tau}) - \mathcal{R}_{\tau} (\bm{y}_{\tau})|| + 
||  \mathcal{R}_{\tau} (\bm{y}_{\tau}) - \bm{y}_{\tau} || = \\
= ||\mathcal{R}_{\tau}(\bm{z}_{\tau} - \bm{y}_{\tau})|| + || \mathcal{R}_{\tau} (\bm{y}_{\tau}) - \bm{y}_{\tau} ||
\leq ||\bm{z}_{\tau} - \bm{y}_{\tau}|| +  ||\mathcal{R}_{\tau} (\bm{y}_{\tau}) - \bm{y}_{\tau}||
= \delta_{\tau} +  ||\mathcal{R}_{\tau} (\bm{y}_{\tau}) - \bm{y}_{\tau}||.
\end{eqnarray*}
Thus we have the recurrence
\begin{equation}\label{eq:deltarecur}
\delta_{\tau+1} \leq \delta_{\tau} +  ||\mathcal{R}_{\tau} (\bm{y}_{\tau}) - \bm{y}_{\tau}||.
\end{equation}
We now use Proposition \ref{prop:claim2} to bound the second term on the right of (\ref{eq:deltarecur}). It follows immediately from the
definition of $\mathcal{R}_{\tau}$ and the fact that the numbers $g_i$ are decreasing with $i$ that
\begin{equation}\label{eq:rdiff}
||\mathcal{R}_{\tau} (\bm{y}_{\tau}) - \bm{y}_{\tau} || \leq \left\{ \begin{array}{lr}
g_{r-3\tau-5}, & {\hbox{if $n$ is even}}, \\ 
g_{r-3\tau-4}, & {\hbox{if $n$ is odd.}} \end{array} \right.
\end{equation}
Hence, for all $\tau \leq m$,
\begin{equation}\label{eq:deltat}
\delta_{\tau} \leq \delta_{\infty} :=
\sum_{k=1}^{\infty} g_{1+3k} = (1-a)\left(\frac{\gamma^{2}}{1-\gamma^{3}} \right) \approx 0.0005.
\end{equation}
Thus for all $\tau \leq m$, and using Proposition \ref{prop:claim1} again,
\begin{equation}\label{eq:triangle}
||\bm{z}_{\tau} - \bm{1}_{\infty}|| \leq ||\bm{z}_{\tau} - \bm{y}_{\tau}|| + 
||\bm{y}_{\tau} - \bm{1}_{\infty}|| \leq \delta_{\infty} + g_1 < \frac{7}{79}.
\end{equation}
Since the operator $\mathcal{R}_{\tau}$ just replaces some elements of $\bm{z}_{\tau}$ with others, we also have $||\mathcal{R}_{\tau}(\bm{z}_{\tau}) - \bm{1}_{\infty}|| < 7/79$ for all $\tau \leq m$. Thus Proposition \ref{prop:claim1} holds for each vector $\mathcal{R}_{\tau}(\bm{z}_{\tau})$. The point is that this is exactly what we need in order to deduce that a finite sequence of $n$ equally spaced agents, with initial gaps of one, will evolve as claimed in Theorem \ref{thm:main}, up to $\tau=m+1$, in other words up to time $t=5(m+1)$, where $m$ is related to $n$ by (\ref{eq:mbound}). This is a direct consequence of observations (a) and (b) above. 
\par To complete the proof of the theorem, it just remains to consider what happens from time $5(m+1)$ onwards. At this time, a total of $6(m+1)$ agents will have become disconnected, and by the next time step will have all collapsed into $2(m+1)$ clusters of size $3$ each. We will be left, at time $5(m+1)$, with a group of somewhere between $5$ and $10$ agents in the middle, depending on the
value of $n \; ({\hbox{mod $6$}})$. The gaps between these remaining agents
will still be less than and close to one, indeed we can use 
a bound similar to (\ref{eq:deltat}). We will need to add the $k=0$ term, but can also start the sum from either $g_1, \, g_2$ or $g_3$, depending on $n \; ({\hbox{mod $6$}})$. Set 
\begin{eqnarray}
\delta_{1,\infty} := \sum_{k=0}^{\infty} g_{1+3k} = (1-a)\left(1 + \gamma + 
\frac{\gamma^{2}}{1-\gamma^{3}} \right) \approx 0.0451, \label{eq:delta1} \\
\delta_{2,\infty} := \sum_{k=0}^{\infty} g_{2+3k} = (1-a)\left(\gamma + 
\frac{\gamma^{3}}{1-\gamma^{3}} \right) + (1-b)(1+\gamma) \approx 0.0114, \label{eq:delta1} \\
\delta_{3,\infty} := \sum_{k=0}^{\infty} g_{3+3k} = (1-a)\left( 
\frac{\gamma}{1-\gamma^{3}} \right) \approx 0.0045. \label{eq:delta1}
\end{eqnarray}
$\;\;$ One can check exhaustively that the remaining middle component of $G_{5(m+1)}$, depending on $n \; ({\hbox{mod $6$}})$, will satisfy the following:
\\
\\
{\sc Case 1:} $n \equiv 5 \, ({\hbox{mod $6$}})$.
\\
\par We have $5$ agents left, with gaps represented by the vector 
$(y_1, \, y_2, \, y_2, \, y_1)$, where
\begin{eqnarray*}
y_1 \geq 1-\delta_{1,\infty} - g_1 \approx 0.9103, \\
y_2 \geq 1-\delta_{1,\infty} - g_2 \approx 0.9435.
\end{eqnarray*}
{\sc Case 2:} $n \equiv 0 \, ({\hbox{mod $6$}})$.
\\
\par We have $6$ agents left, with gaps represented by the vector 
$(y_1, \, y_2, \, y_3, \, y_2, \, y_1)$, where $y_1, \, y_2$ satisfy the
same inequalities as in Case 1 and 
\begin{eqnarray*}
y_3 \geq 1-\delta_{1,\infty} - g_3 \approx 0.9504.
\end{eqnarray*}
{\sc Case 3:} $n \equiv 1 \, ({\hbox{mod $6$}})$.
\\
\par We have $7$ agents left, with gaps represented by the vector 
$(y_1, \, y_2, \, y_3, \, y_3, \, y_2, \, y_1)$, where
\begin{eqnarray*}
y_1 \geq 1-\delta_{2,\infty} - g_1 \approx 0.9440, \\
y_2 \geq 1-\delta_{2,\infty} - g_2 \approx 0.9773, \\
y_3 \geq 1 - \delta_{2,\infty} - g_3 \approx 0.9841.
\end{eqnarray*}
{\sc Case 4:} $n \equiv 2 \, ({\hbox{mod $6$}})$.
\\
\par We have $8$ agents left, with gaps represented by the vector 
$(y_1, \, y_2, \, y_3, \, y_4, \, y_3, \, y_2, \, y_1)$, where
$y_1, \, y_2, \, y_3$ satisfy the same inequalities as in Case 3 and
\begin{eqnarray*}
y_4 \geq 1-\delta_{2,\infty} - g_4 \approx 0.9881.
\end{eqnarray*}
{\sc Case 5:} $n \equiv 3 \, ({\hbox{mod $6$}})$.
\\
\par We have $9$ agents left, with gaps represented by the vector 
$(y_1, \, y_2, \, y_3, \, y_4, \, y_4, \, y_3, \, y_2, \, y_1)$, where
\begin{eqnarray*}
y_1 \geq 1-\delta_{3,\infty} - g_1 \approx 0.9509, \\
y_2 \geq 1-\delta_{3,\infty} - g_2 \approx 0.9842, \\
y_3 \geq 1-\delta_{3,\infty} - g_3 \approx 0.9910, \\
y_4 \geq 1-\delta_{3,\infty} - g_4 \approx 0.9950.
\end{eqnarray*}
{\sc Case 6:} $n \equiv 4 \, ({\hbox{mod $6$}})$.
\\
\par We have $10$ agents left, with gaps represented by the vector 
$(y_1, \, y_2, \, y_3, \, y_4, \, y_5, \, y_4, \, y_3, \, y_2, \, y_1)$, where
$y_1, \, y_2, \, y_3, \, y_4$ satisfy the same inequalities as in Case 5 and
\begin{eqnarray*}
y_5 \geq 1-\delta_{3,\infty} - g_5 \approx 0.9954.
\end{eqnarray*}
When considering the further evolution of one of these six configurations, we can adopt the same strategy as in the proof of Proposition \ref{prop:claim1}. We look on the one hand at what happens when each $y_i$ has the minimum value allowed by the inequalities, on the other at what happens when each $y_i = 1$, and then 
``interpolate'' between these two extremes. We performed all computations in \texttt{Matlab} and it turns out that the following occurs:
\\
\\
{\sc Case 1:} The five agents will always collapse to a single cluster after $6$ time steps, though there are two different possibilities for the sequence of receptivity graphs. One possibility is that agent $3$ will become connected to agents $1$ and $5$ after three steps. If that happens, agents $1$ and $2$ will agree after four steps, as will agents $4$ and $5$. However, these pairs will still be at distance greater than one from one another, so the final collapse to a cluster will require two further time steps. The other possibility is that agent $3$ will become connected to $1$ and $5$ for the first time after four steps. In that case, the merged pair $\{1,\, 2\}$ will be within distance one of the 
merged pair $\{4, \, 5\}$ at step five and thus we collapse to a cluster again at step six. 
\\
\par In all remaining cases, there is only one possible sequence of receptivity graphs. 
\\
\\
{\sc Case 2:} Agents
$1$ and $3$ will become connected after four steps, as will agents $4$ and $6$.
At step five, agents $1, \, 2, \, 3$ will disconnect from $4, \, 5, \, 6$, and
each group of three will collapse to a cluster at step six. 
\\
\\
{\sc Case 3:} At step five, we will split into three components, consisting of
$\{1, \, 2, \, 3\}$, $\{4\}$ and $\{5, \, 6, \, 7\}$. Each boundary component collapses to a cluster at step six.  
\\
\\
{\sc Case 4:} At step five, we will split into three components, consisting of 
$\{1, \, 2, \, 3\}$, $\{4, \, 5\}$ and $\{6, \, 7, \, 8\}$. Each component collapses to a cluster at step six. 
\\
\\
{\sc Case 5:} At step five, we will split into three components, consisting of 
$\{1, \, 2, \, 3\}$, $\{4, \, 5, \, 6\}$ and $\{7, \, 8, \, 9\}$. Each boundary component will collapse to a cluster at step six, whereas the middle component will collapse at step seven.
\\
\\
{\sc Case 6:} At step five, we will split into three components, consisting of 
$\{1, \, 2, \, 3\}$, $\{4, \, 5, \, 6, \, 7\}$ and $\{8, \, 9, \, 10\}$. Each boundary component collapses to a cluster at step six, but the middle component will only collapse at step ten.
\\
\par
The above analysis serves to verify that the values of $\varepsilon(l)$ in (\ref{eq:remainder}) are correct for every $0 \leq l \leq 5$, and thus completes the proof of Theorem \ref {thm:main}.

\setcounter{figure}{0}
\setcounter{table}{0}
\setcounter{equation}{0}

\section{The general case of equally spaced opinions}\label{sect:generald}

In this section we will study the evolution of a finite or infinite sequence of opinions, updating according to (\ref{eq:update}) and initially equally spaced with gaps equal to $d$, where $d$ is an arbitrary real number in the interval $(0, \, 1]$. Theorem \ref{thm:main} describes a striking regularity in the 
evolution when $d=1$. One motivation for looking at other values of $d$ is the observation that there exists some $\epsilon > 0$ such that Theorems \ref{thm:main} and \ref{thm:infmain} hold verbatim for all $d \in (1-\epsilon, \, 1]$. Indeed, this follows immediately from an examination of the argument in Section 2. To begin with, observe that the initial receptivity graph is identical for all $d > 1/2$, namely it is just a chain. Next, let's consider an infinite sequence of agents. Proposition \ref{prop:claim1} tells us that the evolution up to $t = 5$ is identical for all $d > 72/79$. In Proposition \ref{prop:claim2} we just replace $\bm{1}_{\infty}$ by a multiple $d \bm{1}_{\infty}$ of itself and rescale all inequalities - for example, the analogue of (\ref{eq:bound1}) will read: $d(1-a) \leq d - \bm{y}_{\tau}(1) < d(1+\gamma)(1-a)$. Clearly, for $d$ sufficiently close to $1$, such bounds will still be good enough to be able to apply Proposition \ref{prop:claim1}, even allowing for the additional perturbations introduced when returning to finite sequences, which we dealt with at the end of Section \ref{sect:pfmain}. 
\par The obvious question then is to what extent something like Theorems \ref{thm:main} and \ref{thm:infmain} continues to hold as $d$ decreases. 
Figure \ref{fig:ini} 
below suggests that something very interesting may be going on. 
\begin{figure*}[ht!]
  \includegraphics[trim = 20mm 90mm 20mm 70mm, clip,width=1\textwidth]{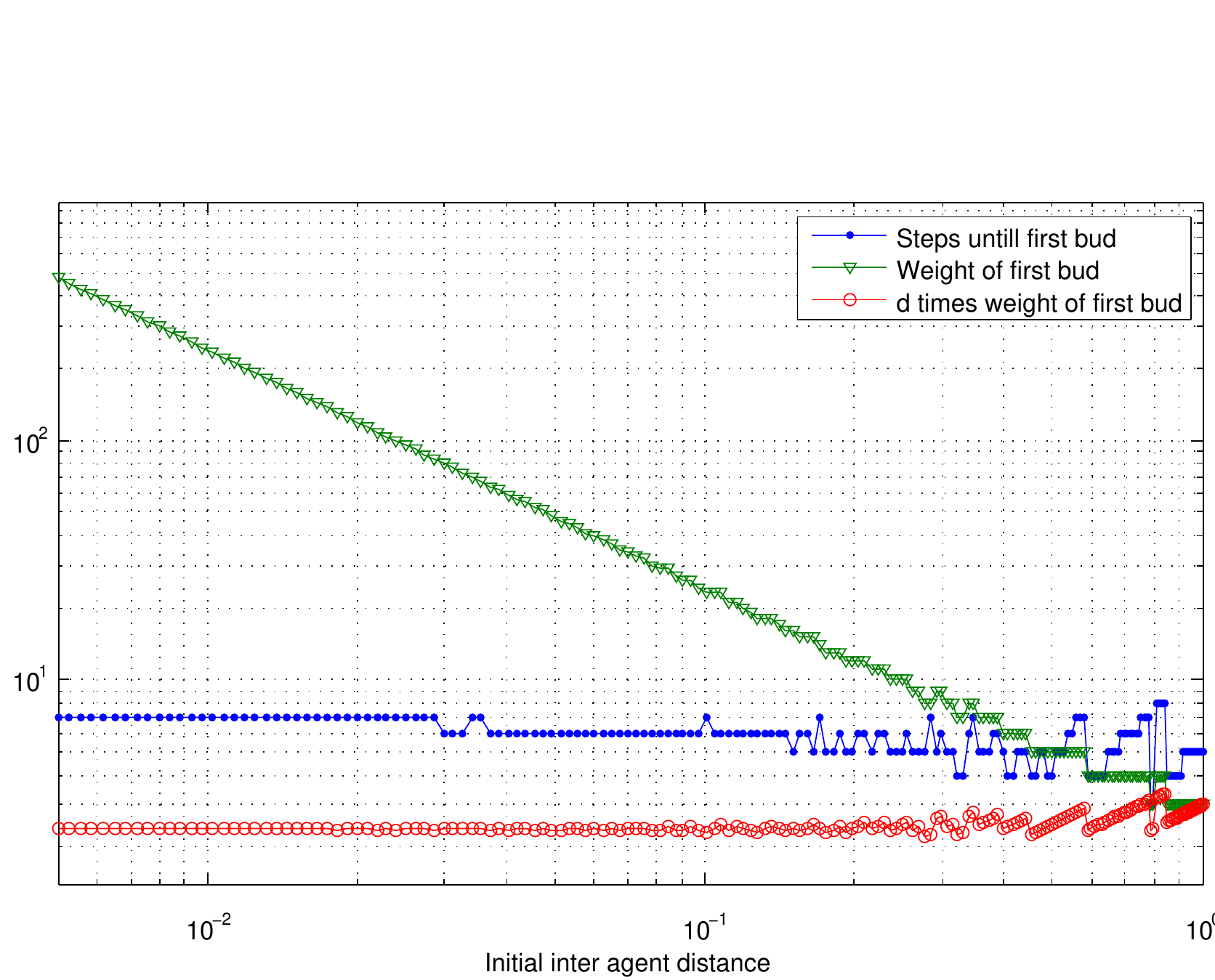} 
\caption{Data for the time of first disconnection.}
\label{fig:ini}
\end{figure*}
On the horizontal axis, we plot $d$ on a logarithmic scale. 
Starting from 1 and decreasing to 0.005, we ran for each value of $d$ a simulation of $\lfloor \frac{30}{d}\rfloor$ agents.
\footnote{The precise expression for the vector of values of $d$ is \texttt{0.005.\^{}(linspace(1,0,100)).\^{}2}
}
This gives the configuration a diameter of $\approx 30$, so that for all times $t < 15$ the edge of the configuration look like the edge of an infinite system with the same $d$.
The blue dots indicate the first time $t = L(d)$ at which the receptivity graph disconnects. The green triangles record the number $M(d)$ of agents that disconnect at time $L(d)$ from each end of the graph. Note that $M(d)$ remains much smaller than $n$ which means that, for every $d$ we simulated, the behaviour up to time $L(d)$ would be identical for any sufficiently large number (depending on $d$) of agents and hence, in particular, for an infinite sequence of agents, except that in that case the graph will only disconnect on the left. This explains our notation. The red circles record the
products $d \cdot M(d)$. 
\par There are two obviously striking features in Figure \ref{fig:ini}. 
Firstly, $L(d)$ does not seem to be increasing as $d$ decreases. It is not constant - recall that Theorem \ref{thm:infmain} says that $L(1) = 5$, while the figure shows that $L(d)$ can attain any value among $\{3,\dots,9\}$. Indeed, all these values are already attained in the interval $d \in \left( \frac{1}{2}, \, 1 \right]$, see Table \ref{tab:ini} below. Moreover, $L(d) = 6$ for ``most'' values of $d$ simulated with $log_{10}(d) \approx -1$, whereas $L(d) = 7$ for most values with $log_{10}(d) \approx -2$. This may suggest logarithmic growth. When $d$ gets really small then simulations become impractical. However, it is obvious to ask

\begin{question}\label{quest:ldbound}
Is there an absolute positive constant $L$ such that $L(d) \leq L$ for all $d \in (0, \, 1]$ ? If not, is it at least true that $L(d) = O \left( \log \frac{1}{d} \right)$ ?
\end{question}

The second striking feature is that the products $d \cdot M(d)$ seem to be
hardly changing $d \rightarrow 0$. In fact, they are all close to $2.38$.
We do not know if they converge to a limit, however.  

\begin{question}\label{quest:dmd} 
Is it true that $d \cdot M(d)$ converges to a limit as $d \rightarrow 0$ ? If so, what is the limit ? If not, is it at least true that $d \cdot M(d) = \Theta(1)$ for all $d \in (0, \, 1]$ ?
\end{question}

These two questions concern only the behaviour up to the first disconnection in the receptivity graph. The meat in Theorems \ref{thm:main} and \ref{thm:infmain} is that this behaviour is then repeated, forever in the case of an infinite sequence of agents and until there are less than $2 \cdot M(1) = 6$ agents left in the finite case. It turns out that such simple periodicity does not hold for arbitrary $d$, though there is a lot of evidence that the behaviour is always ``close'' to periodic. The rest of this section will be concerned with developing this assertion. If we accept it for the moment, then following on from Questions \ref{quest:ldbound} and \ref{quest:dmd} we can ask about the freezing time for an arbitrary configuration of equally spaced agents. If the evolution were perfectly periodic for all $d$ then, as $d \rightarrow 0$, Question \ref{quest:dmd} would imply that the freezing time is $O[L(d) \cdot (dn)]$ for $n \gg_{d} 0$. Notice that $d(n-1)$ is the diameter of the configuration. As Question \ref{quest:ldbound} suggests $L(d)$ grows very slowly, if at all, this would imply that, for general $d$, the diameter is a much better measure of the freezing time than the number of agents. Indeed, a universal bound on $L(d)$ would in turn suggest an affirmative answer to the following: 

\begin{question}\label{quest:freeze}
Does there exist a universal positive constant $\kappa$ such that, for any finite configuration of equally spaced agents obeying (\ref{eq:update}), the freezing time is at most $\kappa \cdot D$, where $D$ is the diameter of the configuration ?
\end{question}

\par To appreciate how far we are from being able to answer any of the questions posed so far, we cannot even prove that, for all $d$ and $n \gg_{d} 0$, the receptivity graph must actually disconnect at all ! Nor can we prove even an $O(n)$ bound for the freezing time, independent of $d$. On the other hand, the intuition that equal spacings of $d=1$ should yield the most slowly converging configuration turns out to be false, even amongst equally spaced configurations. We will prove below that for a short interval of $d$-values slightly above $0.8$, the freezing time is $n + O(1)$. We conjecture that this is the worst-case scenario:

\begin{question}\label{quest:linearbound}
Is it true that any configuration of $n$ equally spaced opinions, which evolve according to (\ref{eq:update}), will freeze by time $n + c$, where $c > 0$ is an absolute constant ?
\end{question}

Let's now go into more detail. In what follows, we 
denote $\bm{\mathcal{E}}_{n, \, d} := (1, \, 1+d, \dots, \, 1+(n-1)d) \in \mathbb{R}^{n}$ and $\bm{\mathcal{E}}_{\infty, \, d} := (1, \, 1+d, \, \dots) \in \mathbb{R}^{\infty}$.
Proposition \ref{prop:claim1} says that the behaviour up to the first disconnection in the receptivity graph is identical for all $d$ in a half-open interval to the left of $d=1$. Basically, the reason for this is that, before the disconnection, edges are only added to the graph, not removed. Consider a fixed edge $\{i,\, j\}$ with $i < j$ and a fixed $d = d^{*}$. If this edge is added to the graph at time $t$ it means that (1) $x_t(j) - x_t (i) \leq 1$,  whereas (2) $x_{t-1}(j) - x_{t-1}(i) > 1$. If we now decrease $d$ and assume that the evolution of the graph is identical up to time $t-1$, then clearly (1) will still hold, while (2) will hold for all $d\in (d^{*} - \varepsilon, \, d^{*}]$ for some $\varepsilon > 0$. Hence, if edges are only added, never deleted, before the first disconnection, then the evolution of the graph during this period will be identical for all $d \in (d^{*} -\varepsilon^{*}, \, d^{*}]$, where $\varepsilon^{*}$ will in general depend on $d^{*}$ - in particular, it will likely be smaller if
the time $L(d^{*})$ at which the disconnection occurs is greater. We may ask whether this is what indeed always happens:

\begin{question}\label{quest:added}
Is it true that, for every $d \in (0, \, 1]$, edges are only added to the receptivity graph, never deleted, before the time $t = L(d)$ at which it disconnects ? Consequently, is it true that there is a decreasing sequence
\begin{equation}\label{eq:dec}
1 = d_0 > d_1 > \cdots
\end{equation}
such that the evolution of the graph up to the first disconnection, and hence the values of the functions $L(d)$ and $M(d)$, are constant on each interval $d \in (d_{i+1}, \, d_i]$ ? Moreover, does the sequence $d_i$ tend to zero ? 
\end{question}

By exhaustive computation we have constructed the sequence (\ref{eq:dec}) down to $d = 1/2$, which is a natural threshold as it marks the point at which the receptivity graph is more than just a chain at $t = 0$. Table \ref{tab:ini} below shows the $12$ different possibilities for the evolution prior to disconnection of an infinite sequence of agents for all $d > 1/2$. The rightmost column in the table describes the precise evolution. Here $\mathcal{A}, \mathcal{B}, \mathcal{C}, \mathcal{D}$ are operators on $l^{\infty}$, considered as bi-infinite matrices. $\mathcal{A}$ is what we called $\widetilde{B}_0$ in Section \ref{sect:pfmain}, the operator corresponding to when the receptivity graph is just a chain. 
$\mathcal{B}$ is what we previously called $\widetilde{B}_4$, corresponding to addition of the edge $\{1,\,3\}$. For $\mathcal{C}$ we in turn add the edges $\{i,\,4\}$, $1 \leq i \leq 2$ and for $\mathcal{D}$ we then add the edges $\{j,\,5\}$, $1 \leq j \leq 3$. $\mathcal{S}$ denotes a shift operator. We take $\bm{y}_0 = d \, \bm{1}_{\infty}$ and the map 
$\bm{y}_0 \mapsto \mathcal{T} \bm{y}_0$ describes the evolution up to time $L(d)$, followed by removing the segment of size $M(d)$ which has become disconnected. 

\begin{table}[ht!]
\begin{center}
\begin{tabular}{|c|c|c|c|c|} \hline 
$i$ & $d_i$ & $L(d_i)$ & $M(d_i)$ & $\mathcal{T} = \mathcal{T}_i$ 
\\ \hline \hline
$0$ & $1$ & $5$ & $3$ & $\mathcal{S}_3 \cdot \mathcal{B} \cdot \mathcal{A}^4$ \\ \hline
$1$ & $\frac{72}{79} \approx 0.9114$ & $4$ & $3$ & $\mathcal{S}_3 \cdot \mathcal{B} \cdot \mathcal{A}^3$ \\ \hline
$2$ & $\frac{864}{1027} \approx 0.8413$ & $9$ & $4$ & $\mathcal{S}_4 \cdot \mathcal{C} \cdot \mathcal{B}^5 \cdot \mathcal{A}^3$ \\ \hline
$3$ & $\frac{31104}{36979} \approx 0.8411$ & $8$ & $4$ & $\mathcal{S}_4 \cdot \mathcal{C} \cdot \mathcal{B}^4 \cdot \mathcal{A}^3$ \\ \hline
$4$ & $\frac{4}{5}$ & $3$ & $3$ & $\mathcal{S}_3 \cdot \mathcal{B} \cdot \mathcal{A}^2$ \\ \hline
$5$ & $\frac{48}{61} \approx 0.7869$ & $7$ & $4$ & $\mathcal{S}_4 \cdot \mathcal{C} \cdot \mathcal{B}^4 \cdot \mathcal{A}^2$ \\ \hline
$6$ & $\frac{1296}{1735} \approx 0.7470$ & $6$ & $4$ & $\mathcal{S}_4 \cdot \mathcal{C} \cdot \mathcal{B}^3 \cdot \mathcal{A}^2$ \\ \hline
$7$ & $\frac{54}{79} \approx 0.6835$ & $5$ & $4$ & $\mathcal{S}_4 \cdot \mathcal{C} \cdot \mathcal{B}^2 \cdot \mathcal{A}^2$ \\ \hline
$8$ & $\frac{2}{3}$ & $5$ & $4$ & $\mathcal{S}_4 \cdot \mathcal{C} \cdot \mathcal{B}^3 \cdot \mathcal{A}$ \\ \hline
$9$ & $\frac{9}{14} \approx 0.6429$ & $4$ & $4$ & $\mathcal{S}_4 \cdot \mathcal{C} \cdot \mathcal{B}^2 \cdot \mathcal{A}$ \\ \hline
$10$ & $\frac{1440}{2459} \approx 0.5856$ & $7$ & $5$ & $\mathcal{S}_5 \cdot \mathcal{D} \cdot \mathcal{C}^3 \cdot \mathcal{B}^2 \cdot \mathcal{A}$ \\ \hline
$11$ & $\frac{69120}{125359} \approx 0.5514$ & $6$ & $5$ & $\mathcal{S}_5 \cdot \mathcal{D} \cdot \mathcal{C}^2 \cdot \mathcal{B}^2 \cdot \mathcal{A}$ \\ \hline
\end{tabular} 
\end{center}
\vspace{0.3cm}
\caption{The possible ``states'' prior to first disconnection for all $d > 1/2$.}
\label{tab:ini}
\end{table}
$\;$ 
To verify that this list is complete can be reduced to  a finite computation. To get a flavour of how this works, consider $d_7 = 54/79$. If 
$\bm{y}_0 := d_7 \bm{1}_{\infty}$ and $\bm{y}_4 := \mathcal{B}^2 \cdot \mathcal{A}^2 \cdot \bm{y}_0$ one may verify that $\bm{y}_{4}(1) = 0$ and $\bm{y}_{4} (2) + \bm{y}_{4} (3) = 1$. This means that the edges $\{1,\,4\}$ and $\{2,\,4\}$ will be added to the graph at $t = 4$ when $d = d_7$ but not when $d > d_7$, so at $d = d_7$ the behaviour changes. If now $\bm{y}_{5} := \mathcal{C} \cdot \bm{y}_4$ then $\frac{1}{d_7} \bm{y}_{5}(4) = \frac{14189}{8640} \approx \frac{1}{0.6089}$, which means that adding the edges $\{1, \, 4\}$ and $\{2, \, 4\}$ at $t = 4$ would result in the removal of edge $\{4, \, 5\}$, and hence a disconnection of the graph, at $t = 5$ for any $d \in \left( \frac{8640}{14189}, \, \frac{54}{79} \right]$. However, this is a faithful description of the evolution only down to $d_8 = 2/3$, because then the behaviour changes already at $t = 1$. For if instead $\bm{y}_0 := d_8 \bm{1}_{\infty}$ and $\bm{y}_1 := \mathcal{A} \cdot \bm{y}_0$ then 
$\bm{y}_{1}(1) + \bm{y}_{1}(2) = 1$ so the edge $\{1, \, 3\}$ is already added 
at $t = 1$. 
\par Now notice from the Table that the quotient $L(d)/M(d)$ is not maximised 
at $d_0 = 1$, it attains greater values of $9/4$ and $8/4 = 2$ at $d_2$ and 
$d_3$ respectively. For values of $d$ in the very narrow interval $(d_3, \, d_2]$,
however, the subsequent evolution is not periodic, rather after a finite amount of time, the behaviour will hop to that exhibited in the interval
$(d_4, \, d_3]$. We will return to this later. However we can prove that, for $d$ lying in some subinterval of $(d_4, \, d_3]$, the behaviour is periodic.

\begin{theorem}\label{thm:slowerd}
There exists a non-empty open subinterval $I \subseteq (d_4, d_3]$ such that for all $d \in I$, the following holds:
\par (i) If the initial configuration $\bm{\mathcal{E}}_{\infty, \, d}$ evolves according to (\ref{eq:update}) then the evolution is periodic: after every eighth time step, a group of four agents disconnect from the left-hand end of the receptivity graph and collapse to a cluster at the subsequent time step.
\par (ii) there are positive constants $C, C^{\prime}$ such that the finite configuration $\bm{\mathcal{E}}_{n, \, d}$ will evolve in an analogous manner, with a group of four agents becoming disconnected at each end after every eighth time step, until there are at most $C$ agents left in the middle. These will collapse to a cluster after at most $C^{\prime}$ further time steps. 
\end{theorem}

Note that the statement in the finite case is slightly less precise than in Theorem \ref{thm:main}. This is because, already for $n = 7$, the evolution of the graph is not constant for $d \in (d_4, \, d_3]$. As may be verified by direct computation, three different things can happen and the freezing time can be $8, \, 9$ or $10$. In any case, the important point about Theorem \ref{thm:slowerd}(ii) is that the freezing time is $n + O(1)$.    

\begin{proof}
The proof is completely analogous to that given in Section \ref{sect:pfmain} so we only present a sketch. Table \ref{tab:ini} yields an immediate analogue of Proposition \ref{prop:claim1}, namely: for an infinite sequence of agents with initial spacings given by $\bm{y}_0 \in l^{\infty}$, if $d_4 \bm{1}_{\infty} < \bm{y}_0 \leq d_3 \bm{1}_{\infty}$, then the evolution up to $t = 8$ is identical to that
described in the $i=3$ row of the table, i.e.: 
\begin{equation*}
\bm{y}_{t} = \mathcal{A}^t \bm{y}_0 \;\; {\hbox{for $0 \leq t \leq 3$}}, \;\;\;\;\;
\bm{y}_{t} = \mathcal{B}^{t-3} \bm{y}_3 \;\; {\hbox{for $3 \leq t \leq 7$}}, \;\;\;\;\;
\bm{y}_8 = \mathcal{C} \bm{y}_7.
\end{equation*}
Secondly, consider the bi-infinite matrix $\mathcal{T}_3$. Analogously to (\ref{eq:Tblock}) and (\ref{eq:Trest}), all the ``action'' takes place in the upper-left
$5 \times 17$ block, and every row from the sixth onwards is just a shift to the right of the previous row. The sum of the entries in every row from the fifth onwards equals one. Let $a,\, b, \, c, \, d$ respectively denote the sums of the entries in the first four rows. One can check that
\begin{equation}\label{eq:newdrowsums}
a_3 = \frac{375281}{373248} \approx 1.0054, \;\; b_3 = \frac{281497}{279936} \approx 1.0056, 
\;\; c = \frac{7787}{7776} \approx 1.0014, \;\; d = \frac{4373}{4374} 
\approx 0.9998.
\end{equation}
One can prove the following analogues of (\ref{eq:bound1})-(\ref{eq:bound3}), the proof reduced as before to solving (in \texttt{Matlab}) a finite collection of polynomial equations:
\\
\par \emph{Let $\bm{y}_0 := \bm{1}_{\infty}$ and for all $\tau \geq 0$, $\bm{y}_{\tau + 1} := \mathcal{T}_{3} \, \bm{y}_{\tau}$. Let $\gamma_{3} := 0.3107$. Then for all $\tau \geq 1$, 
\begin{eqnarray}
a_3 - 1 \leq \bm{y}_{\tau}(1) - 1 \leq (1+\gamma_3) (b_3 - 1), \label{eq:newbound1} \\
b_3 - 1 \leq \bm{y}_{\tau} (2) - 1 \leq (1+\gamma_{3}) (b_3 - 1), \label{eq:newbound2} \\
c-1 \leq \bm{y}_{\tau} (3) - 1 \leq \gamma_3 (b_3 - 1), \label{eq:newbound3} \\
|\bm{y}_{\tau}(i) - 1| \leq \gamma_{3}^{i-2} (b_3 - 1), \;\;\; \forall \, i \geq 4.
\label{eq:newbound4}
\end{eqnarray}}
$\;$ \par 
Note that, according to (\ref{eq:newdrowsums}), there are row sums both above and below one, and thus the sequence $\bm{y}_{\tau}$ is not monotonic as in Proposition \ref{prop:claim2}(i). Thus the left-hand inequalities in (\ref{eq:newbound1}) and (\ref{eq:newbound2}) also contribute to the collection of polynomial equations to be solved here. In addition, because of the lack of monotonicity, it is not immediately obvious that $(\bm{y}_{\tau})$ converges in $l^{\infty}$ to a fixed point of $\mathcal{T}_3$. Though computations suggest this is definitely the case, we have not tried to prove it as we don't need the result. 
\par Now it follows from (\ref{eq:newbound1})-(\ref{eq:newbound4}) that, for all $\tau \geq 0$ and $d \in (d_4 \, d_3]$, 
$||d\bm{1}_{\infty} - d\bm{y}_{\tau}|| \leq$ 
\\ $d_3 (1+\gamma_3)(b_3 - 1) \approx 
0.0061$. This is much less than half the length
of the interval $(d_4, \, d_3]$. So as long as we don't choose $d$ too close to the edges of the interval, the behaviour described in part (i) of Theorem \ref{thm:slowerd} for an infinite sequence of agents has been established. In the case of a finite sequence of $n$ agents there will be additional ``perturbations'' once agents in each half of the sequence affect one another, and these can be analysed by introducing operators $\mathcal{R}_{3,\, \tau}, \, \Theta_{3,\, \tau} = \mathcal{T}_{3} \, \mathcal{R}_{3,\, \tau}$ and vectors $\bm{z}_{\tau} = \bm{z}_{3,\,\tau}$ analogous to (\ref{eq:replace}) and (\ref{eq:yzdefs}). In (\ref{eq:replace}), $\mathcal{R}_{\tau} = \mathcal{R}_{0,\, \tau}$ replaced elements of an input $\bm{x}$ in groups of five and moved this window three steps to the left as $\tau$ increased. In the present case, $\mathcal{R}_{3,\, \tau}$ will replace elements eight at a time and move the window four steps to the left each time. In imitating the argument on page 8, we must exercise a little caution since $||\mathcal{T}_3|| = b_3 > 1$, which in (\ref{eq:Tbort}) would lead to an exponentially growing bound for $\delta_{\tau}$. However, since $\mathcal{T}_{3}$ is essentially a ``band matrix'' and the sum of the entries in each row from the fifth onwards is one, multiplying by $\mathcal{T}_3$ will not magnify errors until $\tau \geq m_3 - O(1)$, in fact until $\tau \geq m_3 - 4$, where now
\begin{equation}\label{eq:m3}
m_3 = \left\{ \begin{array}{lr} \lfloor \frac{r-9}{4} \rfloor, & {\hbox{if $n$ is even}}, \\ \lfloor \frac{r-8}{4} \rfloor, & {\hbox{if $n$ is odd.}} \end{array} \right. \;\;\;\; (r = n/2).    
\end{equation}
Indeed we still have much more margin for error here than in Section \ref{sect:pfmain} so we can continue replacement up to $\tau = m_3 + 1$ and obtain the following analogues of (\ref{eq:triangle}) and (\ref{eq:deltat}):
\begin{equation}\label{eq:newtriangle}
||\bm{z}_{m_3 + 1} - \bm{1}_{\infty}|| \leq \delta^{*}_{\infty} + g^{*}_{1},
\end{equation}
where
\begin{equation}\label{eq:newdeltat}
g^{*}_{1} = d_3 (1+\gamma_3)(b_3 - 1) \approx 0.0061, \;\;\; \delta^{*}_{\infty} = d_3 \cdot b_{3}^{4} \cdot \sum_{k=0}^{\infty} |g^{*}_{1+4k}| = d_3 \cdot b_{3}^{4} \cdot \left[ g^{*}_{1} + (b-1)\left( \frac{\gamma_{3}^{3}}{1-\gamma_{3}^{4}} \right) \right] \approx 0.0054,
\end{equation}
Thus $||\bm{z}_{m_3 + 1} - \bm{1}_{\infty}|| < 0.0116$. 
This is still less than
$\frac{1}{2} (\delta_3 - \delta_4)$, which proves that the evolution of the configuration of $n$ agents up to time $t = 8(m_3 +1)$ will be periodic, at least for all $d$ in some sufficiently small interval around the centre of $(d_4, \, d_3]$. At time $8(m_3 + 1)$ there will be somewhere between $8$ and $15$ agents left, depending on $n \, ({\hbox{mod $8$}})$, so the proof of the theorem is complete.   
\end{proof} 

Theorems \ref{thm:infmain} and \ref{thm:slowerd} prove that, at least for some values of $d > 1/2$, the evolution of the configuration $\bm{\mathcal{E}}_{\infty, \, d}$ is periodic, with $M(d)$ agents becoming disconnected on the left after every $L(d)$ time steps. Such simple periodicity does not hold for arbitrary $d$. Indeed, if $d$ is close to some $d_i$, $i > 0$, and hence close to the boundary between the two intervals $(d_{i}, \, d_{i-1}]$ and $(d_{i+1}, \, d_i]$, then the evolution can jump from one ``state'' to another after a finite time. This is why the ratio $L(d_2)/M(d_2) = 9/4 > 8/4$ does not yield a configuration which converges even more slowly than in Theorem \ref{thm:slowerd}. Recall from Table \ref{tab:ini} that the interval $(d_3, \, d_2]$ is very narrow. One can check that, starting with $d = d_2$, a group of $4$ agents will disconnect at $t = 9$, as in Table \ref{tab:ini}, but at this point the gap between the first two remaining agents will be $\frac{13349725}{15971904} \approx 0.8358 < d_3$, which will suffice for the evolution to jump into a new state whereby the next group of $4$ agents disconnect after $8$ further steps (at $t = 17$). In fact, the system will remain in that state forever, as can be proven by explicit computation and following the proof of Theorem \ref{thm:slowerd}. 
\par In fact, a system can take arbitrarily long to jump from one state to another. To see this, we consider values of $d$ slightly above $d_1 = 72/79$. 
For $n \geq 0$ let $\bm{y}_n := \mathcal{T}_{0}^{n} \bm{1}_{\infty}$, $\bm{z}_n := \mathcal{A}^3 \bm{y}_n$ and 
$d^{n+1} := \frac{1}{\bm{z}_{n}(1) + \bm{z}_{n}(2)}$. By Proposition \ref{prop:claim2}, we know that the sequence $(\bm{y}_n)$ is monotonically decreasing in $l^{\infty}$ and
converging toward a fixed point $\bm{y}_{\infty}$ of $\mathcal{T}_0$. Since, as one can readily check, the matrix $\mathcal{A}^3$ has non-negative entries and row sums at most one, the same is true of the sequence $(\bm{z}_n)$. Hence, $(d^n)$ is an increasing sequence, starting from $d^1 = d_1 = 72/79$ and converging to a limit $d^{\infty} = \frac{1}{\bm{z}_{\infty}(1) + \bm{z}_{\infty}(2)}$, which numerically is about $0.921776...$. It is easy to see that, if $d \in (d^{n-1}, d^n]$, then starting from the configuration $\bm{\mathcal{E}}_{\infty, \, d}$, it will happen $n$ times that $3$ agents disconnect after $5$ time steps, whereas on the $(n+1)$:st occasion, 3 agents will disconnect after $4$ steps instead. Indeed, one can check (though the computations become extremely messy) that the system will then forever remain in the latter state. 
This leads us to our final question:

\begin{question}\label{quest:ultimate}
Is it true that, for any $d \in (0, \, 1]$, the evolution of the system $\bm{\mathcal{E}}_{\infty, \, d}$ is ultimately periodic, in the sense that both the time between successive disconnections and the number of agents which disconnect are constant from some point onwards ? 
\end{question}

Given the ideas introduced in this paper, answering this last question for $d > 1/2$ at least could perhaps be reduced to a finite, if extremely messy computation. However, this would not yield much insight into what happens as $d \rightarrow 0$, which is the main theme behind all the questions posed in this section. The proofs in this paper all relied heavily on explicit numerical estimates (as in (\ref{eq:bound1})-(\ref{eq:bound3}) and (\ref{eq:newbound1})-(\ref{eq:newbound4})). To push the work further, it seems that a deeper qualitative understanding of the evolution of sequences of equally spaced agents and the associated linear operators on $l^{\infty}$ will be needed.   
    


\end{document}